\documentclass[aps,pra,twocolumn,superscriptaddress,amsmath,amssymb,floatfix]{revtex4-2}

\usepackage{graphicx}
\usepackage{dcolumn}
\usepackage{bm}
\usepackage{hyperref}
\usepackage{amsthm}
\usepackage{mathtools}

\hypersetup{
    colorlinks=true,
    linkcolor=blue,
    citecolor=blue,
    urlcolor=blue
}

\newtheorem{theorem}{Theorem}
\newtheorem{lemma}{Lemma}

\begin{document}

\title{Multi-Parameter Error Mitigation in Parametric Down-Conversion via \(\mathrm{SU}(1,1)\) Composite Segmentation}

\author{Nastaran Ragerdi Kashani}
\email{n_ragerdikashani@sbu.ac.ir}
\affiliation{Department of Physics, Shahid Beheshti University, Evin, Tehran, Iran}

\author{Rouhollah Karimzadeh}
\email[Corresponding author: ]{r_karimzadeh@sbu.ac.ir}
\affiliation{Department of Physics, Shahid Beheshti University, Evin, Tehran, Iran}

\date{19 September 2026}

\begin{abstract}
Spontaneous parametric down-conversion is a primary resource for generating non-classical light, yet phase mismatch fluctuations from temperature drift, angular misalignment, and domain fabrication errors reduce pair-generation efficiency and phase stability. Here, we analyze error mitigation in spontaneous parametric down-conversion using composite segmentation governed by $\mathrm{SU}(1,1)$ spatial dynamics. We show that odd spatial Fourier harmonics preserve error-cancellation flatness with pair-generation rates scaling asymptotically as $m^{-2}$, while a coordinate-dilation symmetry provides a design prescription for translating configurations across distinct operating wavelengths ($405$, $532$, and $775\,\mathrm{nm}$). Furthermore, an open-loop thermo-optic bias counteracts off-axis angular drift to broaden the 90\%-efficiency angular acceptance half-width by $3.46\times$, and up to two missing poling domains are tolerated with no numerically resolved loss in relative conversion efficiency. Across a $5^\circ\mathrm{C}$ thermal window, the generated pairs sustain a relative conversion efficiency of $\eta \ge 0.954$ with the modeled phase-coherence metric remaining above $0.95$, providing a structurally robust approach for entangled photon generation without requiring continuous closed-loop feedback stabilization.
\end{abstract}

\maketitle

% =========================================================================
% SECTION I: INTRODUCTION
% =========================================================================
\section{Introduction}
\label{sec:introduction}

Continuous-variable quantum information protocols—such as measurement-based quantum computing on continuous-variable cluster states~\cite{Menicucci2006, Asavanant2019, Larsen2019} and entanglement-based quantum communication~\cite{Ekert1991, Weedbrook2012}—rely on the generation of two-mode squeezed vacuum (TMSV) resource states. A critical requirement for these protocols is maintaining low quadrature noise variance and stable optical phase references~\cite{Braunstein2005}. Because continuous-variable operations and homodyne detections probe specific field quadratures, phase fluctuations induce quadrature rotation, coupling anti-squeezed excess noise into the target measurement quadrature and degrading two-mode quantum correlations.

Spontaneous parametric down-conversion (SPDC) in quadratic ($\chi^{(2)}$) media provides the principal physical mechanism for producing continuous-variable entanglement~\cite{Boyd2020}. In the continuous-wave, undepleted classical pump regime, the spatial propagation dynamics is governed by the non-compact Lie algebra $\mathrm{SU}(1,1)$~\cite{Yurke1986}. In the parametric-gain regime, this non-compact symmetry is naturally represented by hyperbolic rather than bounded periodic trajectories, making the system particularly sensitive to detuning perturbations. Consequently, in standard periodically poled quasi-phase-matched (QPM) structures~\cite{Fejer1992}, the process is acutely sensitive to the wave-vector mismatch $\Delta k$. Ambient thermal fluctuations, angular beam divergence, and fabrication inconsistencies alter the effective phase mismatch, driving the interaction away from nominal phase matching, thereby reducing conversion efficiency and introducing phase drift.

Active temperature stabilization via thermoelectric coolers requires continuous monitoring and packaging volume, which presents a significant overhead for integrated optical setups and field-deployed nodes. Conversely, passive techniques present intrinsic trade-offs: ultra-thin crystals and metasurfaces~\cite{Okoth2019} relax phase-matching constraints but offer low conversion efficiency, demanding intense pump powers that increase unwanted multi-pair emission; chirped or aperiodically poled gratings can broaden the accessible conversion bandwidth~\cite{Chekhova2018}, but generally target spectral engineering rather than explicit dynamical cancellation of parameter drifts; while custom domain engineering~\cite{Branczyk2011, Tambasco2016, Graffitti2018} optimizes joint spectral amplitudes for static configurations without providing dynamical resilience against parameter drift. In quantum control, composite pulse sequences have long been utilized to mitigate systematic detuning errors in $\mathfrak{su}(2)$ systems such as nuclear magnetic resonance~\cite{Levitt1986, Shaka1987, Vandersypen2005}, trapped ions, and superconducting circuits~\cite{Cummins2003, Torosov2011, Merrill2014}. Recently, Erew \emph{et al.}~\cite{Erew2026} adapted this concept to the non-compact Lie algebra $\mathrm{SU}(1,1)$ via detuning-modulated composite segmentation (DMCS), showing that operating within the harmonic dynamical regime enables first-order robustness against detuning errors in a collinear $532\,\mathrm{nm}$ KTP crystal.

While the initial demonstration of DMCS in Ref.~\cite{Erew2026} established the viability of composite segmentation for parametric down-conversion, several practical and theoretical challenges remain unaddressed. The behavior of composite $\mathrm{SU}(1,1)$ sequences under higher-order spatial Fourier harmonics—relevant when fabrication constraints limit the minimum realizable domain-inversion period—has not been explored. Furthermore, in realistic collection geometries involving finite apertures, off-axis emission introduces an asymmetric angular contribution to the phase mismatch that can compromise symmetric error cancellation. Additionally, the portability of composite designs across distinct quantum operational wavelengths requires a rigorous scaling formalism, and the resilience of such segmented sequences to discrete fabrication imperfections, such as missing ferroelectric domains, remains to be quantified from a phase-stability perspective.

In this paper, we address these challenges by developing an error-mitigation framework for $\mathrm{SU}(1,1)$ parametric state preparation. By analyzing the spatial Fourier decomposition of composite gratings, we show analytically (Lemma~1) that odd-order harmonics preserve composite error cancellation to leading order with generation rates scaling asymptotically as $m^{-2}$, while even orders vanish due to parity cancellation. We then model the angular phase mismatch in non-collinear geometries and show that an open-loop thermo-optic control bias can be used to re-center the passband, expanding the operational angular acceptance and protecting single-mode collection efficiency. A coordinate-dilation theorem (Theorem~1) provides a formal scaling prescription establishing that error-canceling trajectories are preserved across different pump-signal configurations, enabling translation across pump wavelengths of $405\,\mathrm{nm}$, $532\,\mathrm{nm}$, and $775\,\mathrm{nm}$ (with degenerate emission near $1550\,\mathrm{nm}$), provided physical parameters are redesigned to satisfy the scaling relations. We also evaluate the impact of discrete domain omissions on the resulting $\mathrm{SU}(1,1)$ evolution, establishing an operational margin within which fabrication defects introduce no measurable efficiency loss. Finally, we benchmark the output states through relative conversion efficiency and modeled phase coherence, demonstrating robust operation across a broad thermal window in comparison with conventional periodic crystals.

% =========================================================================
% SECTION II: ALGEBRAIC FRAMEWORK
% =========================================================================
\section{Algebraic Framework and $\mathrm{SU}(1,1)$ State Space}
\label{sec:algebraic_framework}

We analyze spontaneous parametric down-conversion within a structured quadratic ($\chi^{(2)}$) medium in the undepleted classical pump limit. In the spatial coupled-mode framework along the longitudinal axis $z$, the spatial evolution of the signal and idler mode operators is governed by~\cite{Yurke1986, Boyd2020}:
\begin{align}
\frac{d}{dz} \hat{a}_s(z) &= -i \kappa A_p e^{-i \Delta k z} \hat{a}_i^\dagger(z), \label{eq:spatial_eom_s} \\
\frac{d}{dz} \hat{a}_i^\dagger(z) &= i \kappa A_p e^{i \Delta k z} \hat{a}_s(z), \label{eq:spatial_eom_i}
\end{align}
where $\Delta k = k_p - k_s - k_i$ is the wave-vector mismatch, $A_p$ is the undepleted pump amplitude, and $\kappa$ represents the parametric coupling strength. Setting $\Omega \equiv \kappa A_p$ and defining the slowly varying spatial field operators $\hat{b}_s(z) \equiv \hat{a}_s(z)$ and $\hat{b}_i(z) \equiv \hat{a}_i(z)$, the spatial propagation reads:
\begin{equation}
\frac{d}{dz} \begin{pmatrix} \hat{b}_s(z) \\ \hat{b}_i^\dagger(z) \end{pmatrix} = 
\begin{pmatrix} 
0 & -i \Omega e^{-i \Delta k z} \\ 
i \Omega e^{i \Delta k z} & 0 
\end{pmatrix} 
\begin{pmatrix} \hat{b}_s(z) \\ \hat{b}_i^\dagger(z) \end{pmatrix}.
\label{eq:eom_heisenberg}
\end{equation}
The underlying algebraic structure of Eq.~\eqref{eq:eom_heisenberg} is isomorphic to the non-compact Lie algebra $\mathrm{SU}(1,1)$, spanned by the bilinear generators:
\begin{equation}
\hat{K}_+ = \hat{b}_s^\dagger \hat{b}_i^\dagger, \quad \hat{K}_- = \hat{b}_s \hat{b}_i, \quad \hat{K}_0 = \frac{1}{2} \left( \hat{b}_s^\dagger \hat{b}_s + \hat{b}_i^\dagger \hat{b}_i + 1 \right),
\end{equation}
which fulfill $[\hat{K}_0, \hat{K}_\pm] = \pm \hat{K}_\pm$ and $[\hat{K}_-, \hat{K}_+] = 2 \hat{K}_0$~\cite{Yurke1986}.

To parameterize the spatial trajectory geometrically~\cite{Erew2026}, we introduce the Hermitian observables:
\begin{subequations}
\label{eq:hyperboloid_coords}
\begin{align}
\hat{u}(z) &= \hat{c}_s(z)\hat{c}_i(z) + \hat{c}_s^\dagger(z)\hat{c}_i^\dagger(z), \\
\hat{v}(z) &= i\left[ \hat{c}_s(z)\hat{c}_i(z) - \hat{c}_s^\dagger(z)\hat{c}_i^\dagger(z) \right], \\
\hat{w}(z) &= \hat{c}_s^\dagger(z)\hat{c}_s(z) + \hat{c}_i(z)\hat{c}_i^\dagger(z) - 1,
\end{align}
\end{subequations}
with $\hat{c}_j(z) \equiv \hat{b}_j(z) e^{i \frac{\Delta k}{2} z}$. The expectation vector $\vec{R}(z) = (\langle \hat{u} \rangle, \langle \hat{v} \rangle, \langle \hat{w} \rangle)^T \equiv (u, v, w)^T$ obeys the autonomous precession equation:
\begin{equation}
\frac{d}{dz} \begin{pmatrix} u \\ v \\ w+1 \end{pmatrix} = 
\begin{pmatrix} 
0 & \Delta k & 0 \\ 
-\Delta k & 0 & 2\Omega \\ 
0 & 2\Omega & 0 
\end{pmatrix} 
\begin{pmatrix} u \\ v \\ w+1 \end{pmatrix}.
\label{eq:lorentz_precession}
\end{equation}
Equation~\eqref{eq:lorentz_precession} preserves the indefinite quadratic Casimir metric on the two-sheeted unit hyperboloid:
\begin{equation}
u^2(z) + v^2(z) - [w(z) + 1]^2 = -1.
\label{eq:casimir_invariant}
\end{equation}
The metric $w(z) = 2|\beta(z)|^2$ corresponds to the total mean photon number, representing twice the generated pair count $\mu(z) = |\beta(z)|^2$. The unperturbed vacuum input state corresponds to the point $(u, v, w) = (0, 0, 0)$ on this invariant surface.

In detuning modulated composite segmentation (DMCS)~\cite{Erew2026}, the crystal is partitioned into $N$ discrete domains of lengths $l_j$ and local detunings $\Delta k_j$. The composite Bogoliubov propagator across the assembly is given by:
\begin{equation}
\mathcal{U}^{(N)}(\epsilon) = U_N U_{N-1} \cdots U_2 U_1 = 
\begin{pmatrix} 
\alpha_N(\epsilon) & \beta_N(\epsilon) \\ 
\beta_N^*(\epsilon) & \alpha_N^*(\epsilon) 
\end{pmatrix},
\end{equation}
where $U_j \equiv U_j(\Delta k_j + \epsilon, \Omega_j, l_j)$ represents the propagator of the $j$-th segment.
where $\epsilon$ parameterizes experimental perturbations. Robustness against variations in $\epsilon$ is imposed by setting successive derivatives to zero at the working point:
\begin{equation}
\left. \frac{\partial^m |\beta_N(\epsilon)|^2}{\partial \epsilon^m} \right|_{\epsilon=0} = 0, \quad m \in \{1, 2, \dots, M\}.
\label{eq:flatness_condition}
\end{equation}
Operating within the harmonic bound $(\Delta k_j / 2)^2 > \Omega_j^2$, the generalized frequency $g_j \equiv \sqrt{(\Delta k_j / 2)^2 - \Omega_j^2}$ is strictly real. In this oscillatory regime, the transfer matrix elements of each segment evaluate explicitly to trigonometric functions:
\begin{align}
\alpha_j &= \cos(g_j l_j) + i \frac{\Delta k_j}{2 g_j} \sin(g_j l_j), \label{eq:alpha_trig} \\
\beta_j &= -i \frac{\Omega_j}{g_j} \sin(g_j l_j). \label{eq:beta_trig}
\end{align}
The segment lengths are fixed by the generalized quarter-cycle relation $g_j l_j = \pi / 2$~\cite{Erew2026}:
\begin{equation}
l_j = \frac{\pi}{2 \sqrt{(\Delta k_j / 2)^2 - \Omega_j^2}}.
\label{eq:segment_lengths}
\end{equation}

% =========================================================================
% SECTION III: HIGHER HARMONICS
% =========================================================================
\section{Higher-Order Harmonic Poling and Spatial Parity Selection}
\label{sec:harmonics}

When sub-micron domain widths fall below standard lithographic limits, higher-order quasi-phase-matching (QPM) can be implemented~\cite{Fejer1992}. A spatial square-wave inversion profile with duty cycle $D = 0.5$ admits the Fourier series:
\begin{equation}
\chi^{(2)}(z) = \chi_0^{(2)} \sum_{m=-\infty}^\infty G_m e^{i m K_{\text{pol}} z}, \quad G_m = \frac{2}{m\pi}\sin\left(\frac{m\pi}{2}\right).
\end{equation}
The effective coupling amplitude driving the $m$-th order process is given by:
\begin{equation}
\Omega_m = 
\begin{cases} 
\frac{(-1)^{(m-1)/2}}{m} \Omega_1, & m \in \{1, 3, 5, \dots\} \\ 
0, & m \in \{2, 4, 6, \dots\} 
\end{cases}
\label{eq:omega_parity}
\end{equation}

\begin{figure*}[t]
\centering
\includegraphics[width=\linewidth]{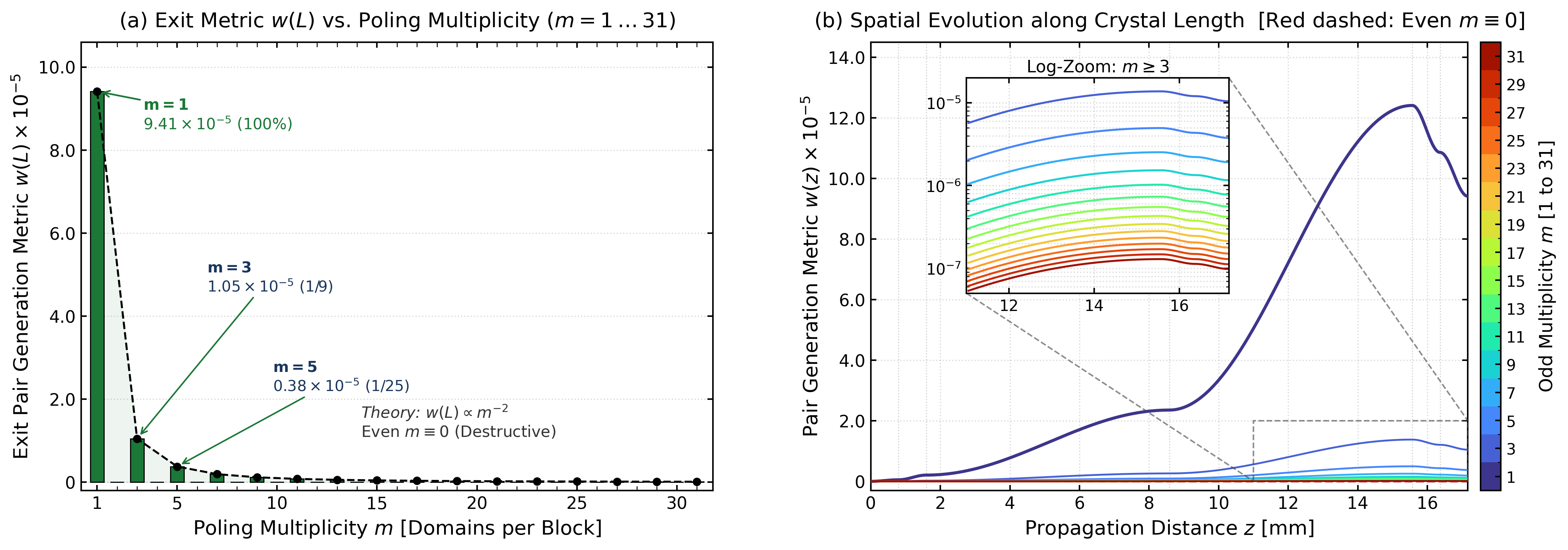}
\caption{(a) Exit pair generation metric $w(L)$ across poling multiplicities $m = 1 \dots 31$. Odd orders approach the predicted asymptotic $m^{-2}$ scaling law (dashed curve). Even orders vanish identically for an ideal symmetric square wave. (b) Trajectory of $w(z)$ along the crystal axis for all sixteen odd harmonics. Inset: logarithmic separation of trajectories for $m \ge 3$.}
\label{fig:w_m}
\end{figure*}

\begin{lemma}[Harmonic Parity Selection and Asymptotic Scaling]
In the low-gain regime ($|\beta_j| \ll 1$) and to leading order in $|\Omega_m| / |\Delta k_j|$, the composite error-nullification conditions defined in Eq.~\eqref{eq:flatness_condition} are approximately preserved for any odd spatial harmonic $m \in \{1, 3, 5, \dots\}$, with the exit pair conversion metric scaling asymptotically as:
\begin{equation}
w(L; m) = \frac{w(L; 1)}{m^2} + \mathcal{O}\left( \frac{\Omega_m^2}{\Delta k_j^2} \right).
\label{eq:m_squared_scaling}
\end{equation}
For the ideal 50\% duty-cycle square-wave profile considered here, the even Fourier coefficients vanish ($G_m = 0$), so isolated even-order contributions are absent ($w(z; m) \equiv 0$).
\end{lemma}

\begin{proof}
In typical low-gain experimental implementations, $|\Delta k_j/2| \gg \Omega_1 \ge |\Omega_m|$. The local generalized frequency $g_j(m) = \sqrt{(\Delta k_j/2)^2 - \Omega_m^2}$ remains approximately equal to $|\Delta k_j|/2$ to within $\mathcal{O}(\Omega_m^2/\Delta k_j^2)$. To this asymptotic order, the quarter-cycle condition $g_j(m) l_j \approx \pi/2$ is preserved across all orders without modifying the nominal segment lengths $l_j$. Evaluating the off-diagonal Bogoliubov element for segment $j$ via Eq.~\eqref{eq:beta_trig} yields $\beta_j(m) \approx -i(\Omega_m/g_j) \sin(g_j l_j)$. In the low-gain regime ($|\beta_j| \ll 1$), the first-order Magnus expansion gives, to leading order, $\beta_N(m) \approx \beta_N(1)/m$. Taking the squared modulus yields $w(L; m) = 2|\beta_N(m)|^2 \approx w(L; 1)/m^2$.
\end{proof}

Figure~\ref{fig:w_m}(a) displays the numerical integration of Eq.~\eqref{eq:lorentz_precession} for multiplicities $m = 1$ through $31$ using the parameters of Design~I in Ref.~\cite{Erew2026}. The terminal value $w(L)$ drops from $9.41 \times 10^{-5}$ ($m=1$) to $1.05 \times 10^{-5}$ ($m=3$) and $0.38 \times 10^{-5}$ ($m=5$), in agreement with Eq.~\eqref{eq:m_squared_scaling}. As evidenced by the spatial trajectories in Fig.~\ref{fig:w_m}(b), the dynamical error-cancellation mechanism remains structurally unperturbed across all odd orders. Because the effective driving parameter scales with pump field $\Omega_m \propto \sqrt{P_p}/m$, the exit pair conversion metric scales linearly with pump power, $w(L) \propto |\beta|^2 \propto P_p / m^2$. Hence, this power-law penalty can be offset by scaling pump power by $P_p' \sim m^2 P_p$, provided the interaction remains safely in the low-gain regime to avoid multipair emissions.

% =========================================================================
% SECTION IV: OFF-AXIS MOMENTUM
% =========================================================================
\section{Off-Axis Acceptance and Thermal Bias Tuning}
\label{sec:off_axis}

In symmetric non-collinear degenerate SPDC with collinear pump, external signal emission at angle $\theta$ introduces an angular phase mismatch:
\begin{equation}
\Delta k(\theta) \approx \frac{4\pi n_s}{\lambda_s}\left( 1 - \sqrt{1 - \frac{\sin^2\theta}{n_s^2}} \right) \approx \frac{2\pi}{\lambda_s n_s}\theta^2.
\label{eq:theta_mismatch}
\end{equation}

\begin{figure*}[t]
\centering
\includegraphics[width=\linewidth]{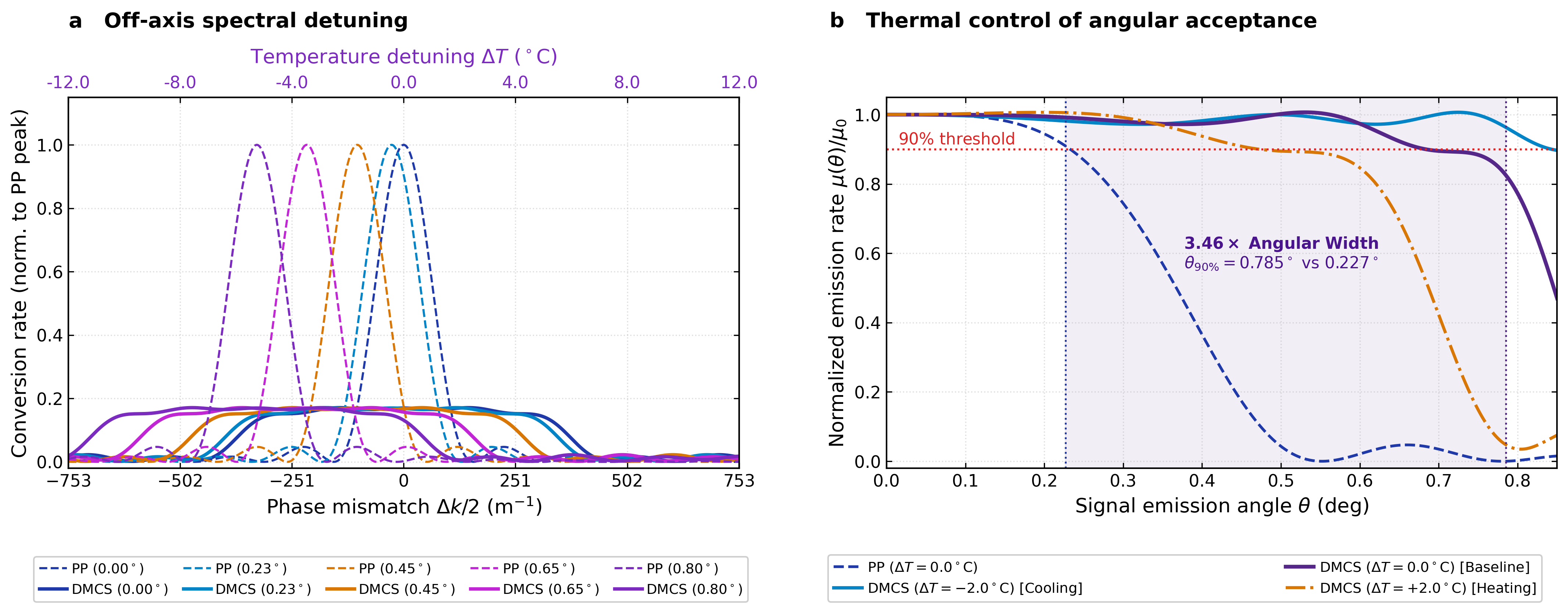}
\caption{(a) Normalized conversion efficiency versus detuning $\Delta k/2$ and temperature for emission angles up to $\theta = 0.80^\circ$ (solid lines: DMCS, dashed lines: PP). (b) Angular acceptance profile at the 90\% threshold. At baseline ($\Delta T = 0^\circ\mathrm{C}$), the DMCS architecture provides an acceptance half-angle of $\theta_{90\%} = 0.785^\circ$, representing a $3.46\times$ enhancement over the PP baseline ($\theta_{90\%}^{\text{PP}} = 0.227^\circ$). Controlled thermal offsets ($\Delta T = -2.0^\circ\mathrm{C}$ cooling, $+2.0^\circ\mathrm{C}$ heating) dynamically steer and broaden the acceptance window.}
\label{fig:thetas}
\end{figure*}

Because $\Delta k(\theta) \ge 0$ is positive-semidefinite, angular divergence introduces an asymmetric, unidirectional drift that breaks the symmetric error cancellation around the collinear working point. To re-center the state trajectory within the flat-top passband, the thermo-optic dispersion $\xi_T \equiv \partial \Delta k / \partial T \approx 125.5\,\mathrm{m}^{-1\circ}\mathrm{C}^{-1}$ acts as an accessible linear control bias:
\begin{equation}
\epsilon_{\text{total}}(\theta, \Delta T) = \Delta k(\theta) + \xi_T \Delta T_{\text{bias}}.
\label{eq:net_drift}
\end{equation}
At baseline ($\Delta T = 0^\circ\mathrm{C}$), as shown in Fig.~\ref{fig:thetas}(b), the composite design achieves a 90\%-efficiency angular acceptance half-width of $\theta_{90\%} = 0.785^\circ$, representing a $3.46$-fold broadening compared to the standard periodically poled benchmark ($\theta_{90\%}^{\text{PP}} = 0.227^\circ$). When collecting at larger off-axis angles, applying a calibrated open-loop negative thermal control bias ($\Delta T = -2.0^\circ\mathrm{C}$) shifts the flat-top passband to offset the positive angular drift, preserving pair-generation efficiency across the collection aperture.

% =========================================================================
% SECTION V: SCALE INVARIANCE
% =========================================================================
\section{Scale Invariance across Wavelength Regimes}
\label{sec:scale_invariance}

The geometry of the $\mathrm{SU}(1,1)$ state trajectories exhibits a scale-invariance symmetry, providing a formal scaling prescription across distinct operating wavelengths via a dimensionless parameter $r > 0$.

\begin{theorem}[Dimensionless Scaling Invariance of the Composite Propagator]
Let $\mathcal{S}_0 \equiv \{l_j, \Delta k_j, \Omega_j\}_{j=1}^N$ denote an $N$-segment sequence satisfying Eq.~\eqref{eq:flatness_condition}. Under the coordinate transformation:
\begin{equation}
l_j' = r \cdot l_j, \quad \Delta k_j' = \frac{\Delta k_j}{r}, \quad \Omega_j' = \frac{\Omega_j}{r}, \quad \epsilon' = \frac{\epsilon}{r},
\label{eq:scaling_transform}
\end{equation}
the scaled composite transfer matrix satisfies:
\begin{equation}
\mathcal{U}^{(N)\prime}(\epsilon') = \mathcal{U}^{(N)}(\epsilon),
\end{equation}
preserving the terminal state coordinates on the hyperboloid and the order of error cancellation.
\end{theorem}

\begin{proof}
Under Eq.~\eqref{eq:scaling_transform}, the accumulated dynamic phase within each segment transforms as $g_j' l_j' = \sqrt{(\Delta k_j'/2)^2 - \Omega_j'^2} (r l_j) = \sqrt{(\Delta k_j/2)^2 - \Omega_j^2} l_j = g_j l_j$. Applying this invariant phase to Eqs.~\eqref{eq:alpha_trig} and~\eqref{eq:beta_trig}, the individual transfer matrix entries satisfy $\alpha_j' = \alpha_j$ and $\beta_j' = \beta_j$. Because each matrix factor in the ordered product is individually invariant, the total propagator $\mathcal{U}^{(N)}$ and its derivatives with respect to the scaled error parameter remain identical.
\end{proof}

\begin{figure*}[t]
\centering
\includegraphics[width=\linewidth]{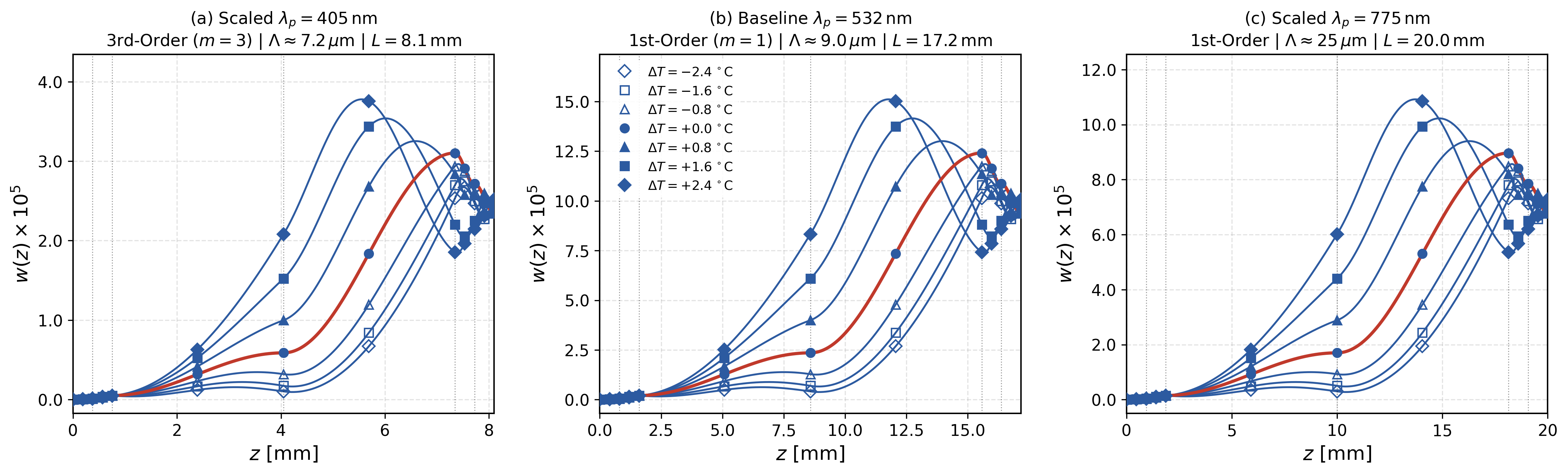}
\caption{\textbf{Scale Invariance Across Operating Wavelength Regimes.} Spatial evolution of the metric $w(z) \times 10^5$ under temperature detunings $\Delta T \in [-2.4^\circ\mathrm{C}, +2.4^\circ\mathrm{C}]$ (red solid curves: unperturbed $\Delta T = 0^\circ\mathrm{C}$, blue curves with markers: detuned states, vertical dotted lines: segment interfaces) across three regimes: (a) Short-wavelength regime ($\lambda_p = 405\,\mathrm{nm} \to 810\,\mathrm{nm}$, $r = 0.472$, $L = 8.10\,\mathrm{mm}$); (b) Reference regime ($\lambda_p = 532\,\mathrm{nm} \to 1064\,\mathrm{nm}$, $r = 1.000$, $L = 17.16\,\mathrm{mm}$); (c) Telecom regime ($\lambda_p = 775\,\mathrm{nm} \to 1550\,\mathrm{nm}$, $r = 1.165$, $L = 20.00\,\mathrm{mm}$). In each case, physical segment lengths and poling periods are redesigned to satisfy the scaled parameters under material dispersion.}
\label{fig:wavelength_scaling}
\end{figure*}

Figure~\ref{fig:wavelength_scaling} demonstrates this scaling relation across three representative physical configurations, where crystal segment lengths and poling periods are designed to match the material dispersion at each target band. For short wavelengths ($\lambda_p = 405\,\mathrm{nm} \to 810\,\mathrm{nm}$), choosing $r = 0.472$ compresses the interaction length to $L = 8.10\,\mathrm{mm}$, while third-order poling ($m = 3$, $\Lambda \approx 7.2\,\mu\mathrm{m}$) avoids sub-micron lithographic constraints. The near-infrared baseline ($\lambda_p = 532\,\mathrm{nm} \to 1064\,\mathrm{nm}$) serves as the unscaled reference ($r = 1.000$, $L = 17.16\,\mathrm{mm}$, $m = 1$, $\Lambda \approx 9.0\,\mu\mathrm{m}$). Finally, for telecommunication wavelengths ($\lambda_p = 775\,\mathrm{nm}$ for degenerate emission near $1550\,\mathrm{nm}$), dilation with $r = 1.165$ scales the sequence to standard commercial formats ($L = 20.00\,\mathrm{mm}$, $m = 1$, $\Lambda \approx 25\,\mu\mathrm{m}$).
Although thermal perturbations induce pronounced trajectory excursions within the crystal bulk, the state vectors exhibit terminal recovery at the exit facet $z = L_{\text{exit}}$, demonstrating that the composite error-cancellation invariants remain preserved under coordinate dilation across distinct operating regimes.

% =========================================================================
% SECTION VI: STRUCTURAL DEFECTS
% =========================================================================
\section{Robustness to Missing Poling Domains}
\label{sec:structural_defects}

Fabrication imperfections during electric-field poling often produce merged or missing ferroelectric domains. We model the cumulative effect of domain omissions as an effective dilation in the average domain period over a crystal of fixed total length $L$. This effective model captures the cumulative phase-mismatch shift associated with domain-count errors over the fixed device length, without resolving the local spatial profile of an individual missing reversal. For a device of nominal domain count $M_{\text{total}}$, omitting $\Delta M$ domain reversals dilates the average domain period to $\bar{\Lambda} \approx \Lambda_0 (1 + \Delta M / M_{\text{total}})$. Consequently, the grating reciprocal wavevector is shifted by $\delta K_{\text{pol}} = \frac{2\pi}{\bar{\Lambda}} - K_{\text{pol}}^{(0)} \approx - (\frac{\Delta M}{M_{\text{total}}}) K_{\text{pol}}^{(0)}$. Because the net phase mismatch is $\Delta k = (k_p - k_s - k_i) - K_{\text{pol}}$, this reduction in grating momentum introduces an effective positive shift in the phase mismatch:
\begin{equation}
\Delta k_{\text{shift}} = - \delta K_{\text{pol}} = + \left( \frac{\Delta M}{M_{\text{total}}} \right) K_{\text{pol}}^{(0)}.
\label{eq:dk_defect}
\end{equation}

\begin{figure*}[t]
\centering
\includegraphics[width=\linewidth]{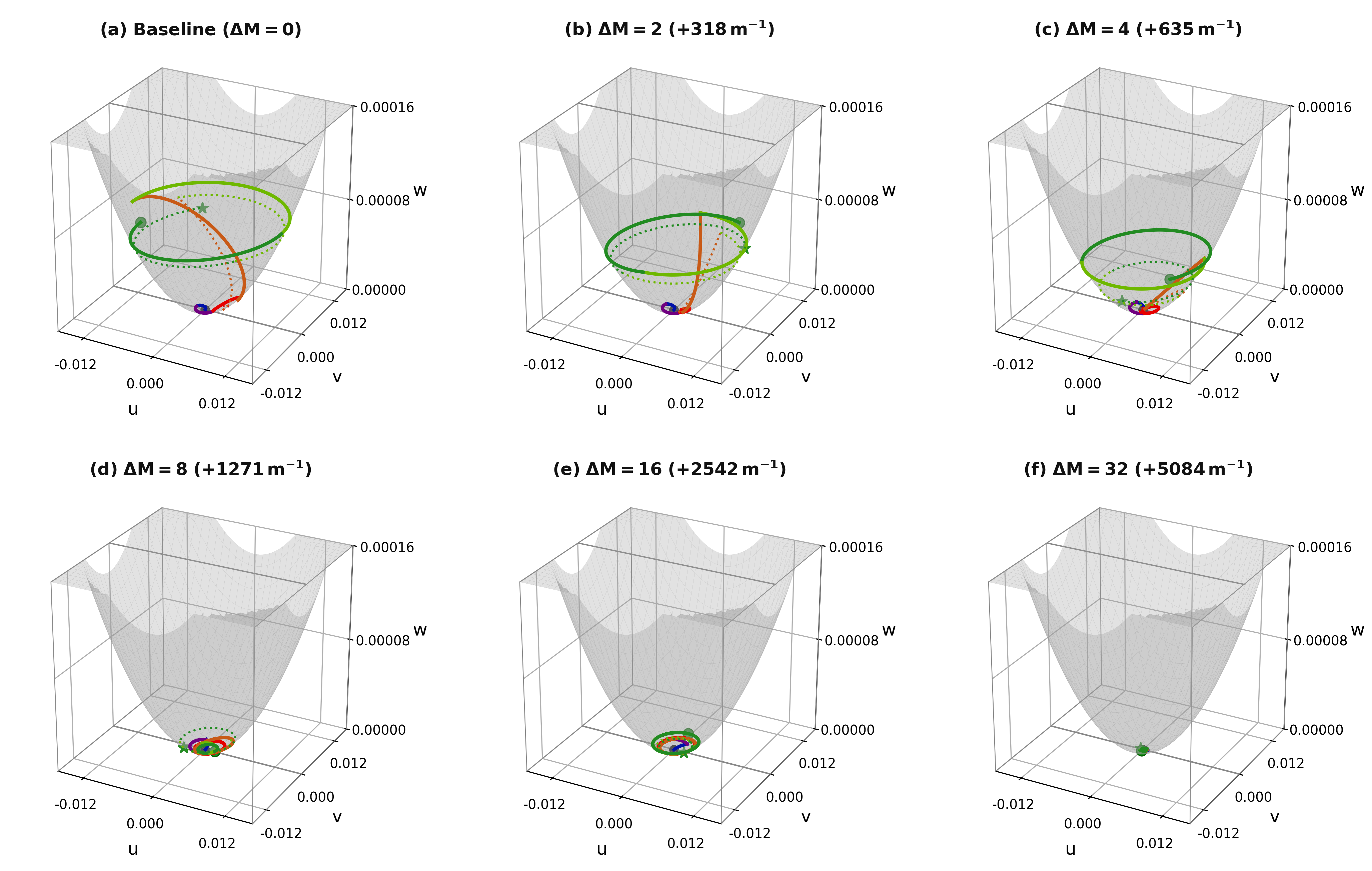}
\caption{\textbf{Geometric Tolerance Against Missing Poling Domains on the $\mathrm{SU}(1,1)$ Hyperboloid.} Three-dimensional trajectories of the state vector $(u, v, w)$ across six structural defect scenarios with missing domain counts $\Delta M \in \{0, 2, 4, 8, 16, 32\}$ (out of $M_{\text{total}} = 4444$). Solid curves denote unperturbed trajectories ($\Delta T = 0^\circ\mathrm{C}$), while dotted curves illustrate the evolution under thermal perturbation ($\Delta T = +1.0^\circ\mathrm{C}$). Cases~1 and 2 reside entirely within the flat-top band ($w = 6.4 \times 10^{-5}$, $100\%$ relative efficiency); Case~3 operates at the passband edge ($83\%$ efficiency); in Cases~4--6, the simulated conversion is strongly suppressed, keeping the state localized near the vacuum point $(u, v, w) = (0, 0, 0)$.}
\label{fig:hyperboloid_defects}
\end{figure*}

Figure~\ref{fig:hyperboloid_defects} tracks the state trajectory on the hyperboloid across six defect scenarios in a $20\,\mathrm{mm}$ crystal ($M_{\text{total}} = 4444$). In the simulated defect configurations, omission of up to $\Delta M = 2$ domains ($\Delta k_{\text{shift}} = +318\,\mathrm{m}^{-1}$) produced no numerically resolved reduction in the normalized conversion metric ($w = 6.4 \times 10^{-5}$, $100\%$ relative efficiency). At $\Delta M = 4$ ($\Delta k_{\text{shift}} = +635\,\mathrm{m}^{-1}$), the state reaches the passband edge, retaining $83\%$ of its peak value ($w = 5.3 \times 10^{-5}$). For severe defects ($\Delta M \ge 8$, $\Delta k_{\text{shift}} \ge +1271\,\mathrm{m}^{-1}$), the simulated conversion is strongly suppressed, keeping the state localized near the vacuum point $(u, v, w) = (0, 0, 0)$.
Thus, the flat-top passband accommodates up to $\Delta M = 2$ missing domains without efficiency loss, setting a practical threshold for fabrication defects.

% =========================================================================
% SECTION VII: QUANTUM METRICS
% =========================================================================
\section{Quantum State Robustness and Phase Coherence Metrics}
\label{sec:quantum_metrics}

\subsection{Phase Coherence and Squeezing Noise Leakage}
In the continuous-variable picture, the Bogoliubov transformation $\mathcal{U}^{(N)}$ generates a two-mode squeezed vacuum state $|\psi(\epsilon)\rangle = \frac{1}{|\alpha(\epsilon)|} \sum_{n=0}^\infty \lambda^n(\epsilon) |n, n\rangle$, where $\lambda(\epsilon) \equiv -\beta(\epsilon)/\alpha^*(\epsilon)$. The relative pair-generation efficiency with respect to the ideal unperturbed setting is given by:
\begin{equation}
\eta(\epsilon) \equiv \frac{w(\epsilon)}{w(0)} = \frac{|\beta(\epsilon)|^2}{|\beta(0)|^2}.
\label{eq:efficiency_metric}
\end{equation}
In continuous-variable architectures, homodyne detection probes field quadratures where phase fluctuations induce quadrature rotation, coupling anti-squeezed variance into the squeezed quadrature:
\begin{equation}
\Delta^2 \hat{x}_-(\Delta\phi) \approx \frac{1}{2} \left[ e^{-2r(\epsilon)} + \Delta\phi^2 e^{2r(\epsilon)} \right], \quad (|\Delta\phi| \ll 1)
\label{eq:quadrature_mixing}
\end{equation}
where $r(\epsilon) = \sinh^{-1}(|\beta(\epsilon)|)$ denotes the squeezing parameter, and $\Delta \phi(\epsilon) \equiv \arg[\beta(\epsilon)] - \arg[\beta(0)]$ represents the relative phase drift. 

To quantify the phase stability of the parametric process, we define the modeled phase-coherence metric:
\begin{equation}
C_\phi(\epsilon) \equiv \cos^2\left( \Delta \phi(\epsilon) \right).
\label{eq:coherence_metric}
\end{equation}
Here, $C_\phi$ serves as a modeled phase-alignment metric reflecting the interference fringe contrast relative to the unperturbed reference state. Within the ideal pure two-mode squeezed state representation, the bipartite entanglement is characterized by the logarithmic negativity $E_N(\epsilon) = 2 r(\epsilon) / \ln 2 \approx 2|\beta(\epsilon)|/\ln 2$, which scales with the conversion amplitude as $E_N(\epsilon) / E_N(0) \approx \sqrt{\eta(\epsilon)}$. The resulting robustness therefore preserves the conversion efficiency, while the modeled phase-alignment metric $C_\phi$ provides an independent characterization of the accumulated phase stability.

\subsection{Comparative Analysis of Operational Scenarios}
Table~\ref{tab:quantum_metrics} quantitatively benchmarks the generated states across five operational scenarios. In the standard periodically poled crystal, introducing an off-axis angle of $\theta = 0.45^\circ$ or omitting $\Delta M = 2$ domain stripes detunes the phase matching well beyond the narrow acceptance bandwidth, leading to state extinction ($\eta_{\text{PP}} \to 0$), where the phase coherence metric becomes physically undefined ($\text{N/A}$).

In contrast, the composite DMCS architecture demonstrates structural resilience across all tested perturbations. Under unsteered off-axis emission (Scenario~2), the state remains within the flat-top band, sustaining $\eta_{\text{DMCS}} = 0.990$ and a modeled phase coherence of $C_{\phi,\text{DMCS}} = 0.995$. Applying the open-loop thermal bias ($\Delta T = -2^\circ\mathrm{C}$, Scenario~3) recenters the state trajectory on the $\mathrm{SU}(1,1)$ hyperboloid, restoring near-unity efficiency ($\eta = 0.999$) and maximal phase alignment ($C_\phi = 1.000$). Furthermore, under stochastic fabrication errors, the DMCS design absorbs up to two missing domains without performance loss (Scenario~4) and maintains robust phase alignment ($C_\phi = 0.955$) even under four missing domains (Scenario~5). Over a pure thermal sweep of $\Delta T \in [-2.5^\circ\mathrm{C}, +2.5^\circ\mathrm{C}]$ at $\theta=0$, the design sustains $\eta \ge 0.954$ and $C_\phi \ge 0.954$, indicating robust preservation of the modeled conversion efficiency and phase-alignment metrics under the considered perturbations.
\vspace{1em}
\begin{table}[tb]
\caption{\textbf{Benchmarking Operational Scenarios Against Standard PP and Nominal DMCS.} Comparison between the benchmark 20~mm periodically poled (PP) crystal and the composite DMCS design (Design~VI, scaled to $L=20\,\mathrm{mm}$) across operational scenarios. $\eta \equiv w(\epsilon)/w(0)$ denotes the relative pair conversion metric, and $C_\phi \equiv \cos^2(\Delta\phi)$ represents the modeled phase-coherence metric.}
\vspace{0.5em}
\label{tab:quantum_metrics}
\resizebox{\columnwidth}{!}{%
\begin{tabular}{lcccc}
\hline\hline
Operational Scenario & $\eta_{\text{PP}}$ & $C_{\phi,\text{PP}}$ & $\eta_{\text{DMCS}}$ & $C_{\phi,\text{DMCS}}$ \\
\hline
1. Baseline ($\theta=0^\circ, \Delta M=0$) & $1.000$ & $1.000$ & $1.000$ & $1.000$ \\
2. Off-axis ($\theta=0.45^\circ$, unsteered) & $0.175$ & $0.046$ & $0.990$ & $0.995$ \\
3. Off-axis ($\theta=0.45^\circ$, steered $\Delta T=-2^\circ\mathrm{C}$) & $0.000$ & \text{N/A} & $\mathbf{0.999}$ & $\mathbf{1.000}$ \\
4. Defect ($\Delta M=2$ missing domains) & $0.000$ & \text{N/A} & $\mathbf{1.000}$ & $\mathbf{1.000}$ \\
5. Defect ($\Delta M=4$ missing domains) & $0.000$ & \text{N/A} & $\mathbf{0.830}$ & $\mathbf{0.955}$ \\
\hline\hline
\end{tabular}%
}
\end{table}
\vspace{-2em}
\section{Conclusion and Outlook}
\label{sec:conclusion}

We have established a geometric framework for error mitigation in continuous-variable parametric down-conversion governed by the non-compact Lie algebra $\mathrm{SU}(1,1)$. Operating in the harmonic regime suppresses the sensitivity of the output pair-generation metric to thermal drifts, transverse angular misalignments, and discrete fabrication errors. This composite error-mitigation framework sustains a relative conversion efficiency of $\eta \ge 0.954$ and a modeled phase coherence of $C_\phi \ge 0.954$ across a $5^\circ\mathrm{C}$ thermal sweep without requiring continuous closed-loop feedback stabilization.

Furthermore, the asymptotic $m^{-2}$ scaling law derived in Lemma~1 demonstrates that higher spatial Fourier harmonics relax sub-micron fabrication constraints, where the power-law penalty can be offset by pump intensity in the undepleted regime. Beyond nonlinear optical crystals, the scaling transformation also applies mathematically to other non-compact bosonic parametric models, such as traveling-wave parametric amplifiers in superconducting circuits, provided that coupling and detuning parameters obey the same dimensionless scaling relations.


\begin{thebibliography}{99}

\bibitem{Menicucci2006}
N. C. Menicucci, P. van Loock, M. Gu, C. Weedbrook, T. C. Ralph, and M. A. Nielsen, \emph{Universal Quantum Computation with Continuous-Variable Cluster States}, Phys. Rev. Lett. \textbf{97}, 110501 (2006).

\bibitem{Asavanant2019}
W. Asavanant, Y. Shiozawa, S. Yokoyama, B. Charoensombutamon, H. Kouhei, S. Takeda, J. Yoshikawa, P. van Loock, and A. Furusawa, \emph{Generation of time-domain-multiplexed two-dimensional cluster state}, Science \textbf{366}, 373 (2019).

\bibitem{Larsen2019}
M. V. Larsen, X. Guo, C. R. Breum, J. S. Neergaard-Nielsen, and U. L. Andersen, \emph{Deterministic generation of a two-dimensional cluster state}, Science \textbf{366}, 369 (2019).

\bibitem{Ekert1991}
A. K. Ekert, \emph{Quantum cryptography based on Bell's theorem}, Phys. Rev. Lett. \textbf{67}, 661 (1991).

\bibitem{Weedbrook2012}
C. Weedbrook, S. Pirandola, R. Garc\'ia-Patr\'on, N. J. Cerf, T. C. Ralph, J. H. Shapiro, and S. Lloyd, \emph{Gaussian quantum information}, Rev. Mod. Phys. \textbf{84}, 621 (2012).

\bibitem{Braunstein2005}
S. L. Braunstein and P. van Loock, \emph{Quantum information with continuous variables}, Rev. Mod. Phys. \textbf{77}, 513 (2005).

\bibitem{Boyd2020}
R. W. Boyd, \emph{Nonlinear Optics}, 4th ed. (Academic Press, London, 2020).

\bibitem{Yurke1986}
B. Yurke, S. L. McCall, and J. R. Klauder, \emph{$SU(2)$ and $SU(1,1)$ interferometers}, Phys. Rev. A \textbf{33}, 4033 (1986).

\bibitem{Fejer1992}
M. M. Fejer, G. A. Magel, D. H. Jundt, and R. L. Byer, \emph{Quasi-phase-matched second harmonic generation: tuning and tolerances}, IEEE J. Quantum Electron. \textbf{28}, 2631 (1992).

\bibitem{Levitt1986}
M. H. Levitt, \emph{Composite pulses}, Prog. Nucl. Magn. Reson. Spectrosc. \textbf{18}, 61 (1986).

\bibitem{Shaka1987}
A. J. Shaka and A. Pines, \emph{Symmetric phase-alternating composite pulses}, J. Magn. Reson. \textbf{71}, 495 (1987).

\bibitem{Vandersypen2005}
L. M. K. Vandersypen and I. L. Chuang, \emph{NMR techniques for quantum control and computation}, Rev. Mod. Phys. \textbf{76}, 1037 (2005).

\bibitem{Cummins2003}
H. K. Cummins, G. Llewellyn, and J. A. Jones, \emph{Tackling systematic errors in quantum logic with composite pulses}, Phys. Rev. A \textbf{67}, 042308 (2003).

\bibitem{Torosov2011}
B. T. Torosov and N. V. Vitanov, \emph{Smooth composite pulses for high-fidelity quantum control}, Phys. Rev. A \textbf{83}, 053420 (2011).

\bibitem{Merrill2014}
J. T. Merrill and K. R. Brown, \emph{Progress in compensating pulse sequences for quantum computing}, Adv. Chem. Phys. \textbf{154}, 241 (2014).

\bibitem{Okoth2019}
C. Okoth, A. Cavanna, T. Santiago-Cruz, and M. Chekhova, \emph{Microscale Generation of Entangled Photons without Momentum Conservation}, Phys. Rev. Lett. \textbf{123}, 263602 (2019).

\bibitem{Chekhova2018}
M. V. Chekhova, S. Germanskiy, D. B. Horoshko, G. K. Kitaeva, M. I. Kolobov, G. Leuchs, C. R. Phillips, and P. A. Prudkovskii, \emph{Broadband bright twin beams and their upconversion}, Opt. Lett. \textbf{43}, 375 (2018).

\bibitem{Branczyk2011}
A. M. Bra\'nczyk, A. Fedrizzi, T. M. Stace, T. C. Ralph, and A. G. White, \emph{Engineered optical nonlinearity for quantum light sources}, Opt. Express \textbf{19}, 55 (2011).

\bibitem{Tambasco2016}
J.-L. Tambasco, A. Boes, L. G. Helt, M. J. Steel, and A. Mitchell, \emph{Domain engineering algorithm for practical and effective photon sources}, Opt. Express \textbf{24}, 19616 (2016).

\bibitem{Graffitti2018}
F. Graffitti, P. Barrow, M. Proietti, D. Kundys, and A. Fedrizzi, \emph{Independent engineering of joint spectral amplitude and phase in spontaneous parametric down-conversion}, Optica \textbf{5}, 514 (2018).

\bibitem{Erew2026}
M. Erew, Y. Reches, O. Yesharim, M. Goldstein, A. Arie, and H. Suchowski, \emph{Efficient Robust Spontaneous Parametric Down-Conversion via Detuning Modulated Composite Segments Designs}, Laser \& Photonics Reviews \textbf{20}, e00947 (2026).

\end{thebibliography}
\end{document}